\documentclass[conference, 10pt]{IEEEtran}
\IEEEoverridecommandlockouts
\usepackage{cite}
\usepackage{amsmath,amssymb,amsfonts}
\usepackage{graphicx}
\usepackage{textcomp}
\usepackage{xcolor}
\usepackage{subcaption}
\usepackage{bm}
\usepackage{multirow}
\usepackage{booktabs}
\usepackage{subcaption}
\usepackage{amsthm}\usepackage{balance}\usepackage{url}
\usepackage{algorithm,algpseudocode}
\usepackage{amssymb}
\usepackage{nccmath}
\usepackage{makecell}

\newtheorem{theorem}{Theorem}
\newtheorem*{proof*}{Proof}
\newtheorem{proposition}{Proposition}

\algrenewcommand\algorithmicrequire{\textbf{Initialization:}}
\algrenewcommand\algorithmicensure{\textbf{Output:}}
\algdef{SE}[DOWHILE]{Do}{doWhile}{\algorithmicdo}[1]{\algorithmicwhile\ #1}%

\makeatletter

\newcommand{\Rmnum}[1]{\expandafter\@slowromancap\romannumeral #1@}
\makeatother

\date{}

\def\BibTeX{{\rm B\kern-.05em{\sc i\kern-.025em b}\kern-.08em
    T\kern-.1667em\lower.7ex\hbox{E}\kern-.125emX}}

\DeclareMathSizes{10}{9}{7}{5}

\begin{document}
\title{\huge{Distortion-Aware UAV Placement for Aerial Semantic Relay Communications: An Analytical Approach}
}

\author{Mingze Gong\textsuperscript{1}, Jia Yan\textsuperscript{1}, and Shuoyao Wang\textsuperscript{2}   \\
    \textsuperscript{1} Intelligent Transportation Thrust, The Hong Kong University of Science and Technology (Guangzhou), China\\
    \textsuperscript{2} College of Electronic and Information Engineering, Shenzhen University, China\\
    E-mail: mgong256@connect.hkust-gz.edu.cn, jasonjiayan@hkust-gz.edu.cn, sywang@szu.edu.cn
    }
\maketitle

\begin{abstract}\label{abstract}
    Aerial semantic relay communications (SRC) employs an unmanned aerial vehicle (UAV) equipped with a semantic encoder as a relay, which not only extends the data acquisition coverage of the base station (BS) from resource-limited sensing device (SD) but also enhances communication efficiency through semantic feature transmission over the UAV-BS link.
    Existing works mainly focus on sum-rate maximization, overlooking the end-to-end reconstruction distortion of sensory data in UAV-assisted SRC systems.
    Optimizing the UAV placement is crucial for minimizing the end-to-end reconstruction distortion, as it fundamentally trades off the input perturbation at the UAV-side encoder against that at the BS-side decoder through the two-hop wireless channel conditions. In this paper, we propose an interpretable and efficient UAV placement policy by minimizing end-to-end reconstruction distortion in aerial SRC.
    This is a challenging task since the black-box nature of the DNN-based codecs and the intricate coupling between the heterogeneous codec sensitivities, along with two-hop channel impairments, render the end-to-end distortion analytically intractable to characterize.
    We first derive an analytical expression of the end-to-end distortion, explicitly revealing the impact of cross-hop perturbation coupling, wireless channel and radio resource on the reconstruction error. 
    Based on that, we develop a closed-form UAV placement strategy with fast adaptability across various aerial SRC system configurations.
    Numerical results demonstrate that the proposed distortion-aware UAV deployment closely tracks the empirical exhaustive-search optimum, while achieving lower distortion compared to representative capacity-based and curve-fitting benchmarks.
\end{abstract}

\section{Introduction}\label{introduction}
    Driven by deep neural networks (DNNs) and joint source-channel coding (JSCC), semantic communication has shifted from bit-level transmission to semantic-level fidelity, achieving remarkable communication efficiency~\cite{9955525, 8723589}.
    However, resource-constrained sensing devices (SDs) often cannot perform DNN-based semantic encoding when transmitting sensory data to a remote base station (BS).
    To address this limitation, aerial semantic relay communication (SRC)~\cite{10638143, 10560514} offloads semantic encoding to an unmanned aerial vehicle (UAV), which extracts semantic features from raw sensory data and relays them to the BS-side decoder for reconstruction.
    This framework extends data acquisition coverage via cooperative communication~\cite{9298816} and improves spectrum efficiency through the UAV-BS semantic link.

    The position of the UAV needs to be carefully determined to improve the performance of the aerial SRC system. 
    Specifically, \cite{11224400} studied a secure aerial SRC system to maximize the secrecy semantic data rate while minimizing delay and energy consumption.
    In addition, \cite{yin2025aerial} considered the aerial semantic relay-enabled space-air-ground integrated networks, where the transmit power, bandwidth, and UAV positions are jointly optimized by maximizing the system's sum-rate.

    However, existing works~\cite{10638143, 11224400, yin2025aerial} mainly focus on sum-rate maximization, while overlooking the end-to-end reconstruction distortion of sensory data in UAV-assisted SRC systems.
    Specifically, the raw sensory data received at the UAV from the SD is first corrupted by the wireless channel fading and communication noise of  the SD-UAV link before being fed into the onboard semantic encoder. The resulting semantic features output by
the encoder are then transmitted to the BS over an imperfect
UAV-BS wireless link, incurring further perturbation.
    Finally, the perturbed semantic features received at the BS are input to the semantic decoder to reconstruct the original sensory
data. 
    Such end-to-end reconstruction error critically depends on the UAV placement. 
    In particular, when the UAV hovers
closer to the SD, one can attain a less-perturbed input to the semantic encoder at the cost of a degraded UAV-BS channel,
which increases the perturbation of the semantic features input
to the BS-side decoder. Conversely, when the UAV moves
closer to the BS, the decoder receives higher-quality semantic
features, but the increased SD-UAV distance amplifies the
perturbation at the encoder input.
    This fundamental trade-off necessitates a judicious UAV placement to minimize the end-to-end reconstruction distortion.


    Tackling this task is challenging because the end-to-end distortion in aerial SRC systems is analytically intractable to characterize.
    First, the highly nonlinear mappings learned by the DNN-based codecs operate as black boxes, making it inherently difficult to derive closed-form expressions for their sensitivity to input perturbations. 
    Second, the encoder and decoder respond differently to their respective input perturbations \cite{10328181}, and these heterogeneous sensitivities are intricately intertwined in shaping the overall reconstruction error. 
    More specifically, the decoder input is affected not only by the UAV-BS channel, but also by the encoder-transformed source perturbation induced by the SD-UAV link.
    This cascaded interaction between wireless channels and codec sensitivities makes distortion analysis highly nontrivial.

    Notice that recent works~\cite{10189867, zhang2026unanticipated} have analyzed end-to-end distortion under single-hop perturbations of semantic features, while ignoring the intricate coupling between the input perturbations at both the UAV-side encoder and the BS-side decoder. 
    This oversight renders their distortion analysis inapplicable to the considered aerial SRC system.
    Alternatively, one can resort to curve fitting based on empirical measurements to model the end-to-end distortion \cite{11006481}.
    However, beyond offering limited engineering insights, such data-driven approaches typically require re-fitting with newly collected samples whenever the system configuration changes, which can be prohibitively time-consuming in dynamic network environments.
    In a nutshell, \emph{an analytical framework for end-to-end distortion characterization under coupled two-hop perturbations is urgently needed to guide efficient UAV placement in aerial SRC systems.}

    In this paper, we aim to design interpretable and efficient UAV placement policy that minimizes end-to-end reconstruction distortion in aerial SRC systems.
    The main contributions are summarized as follows:
    \begin{itemize}
        \item \emph{End-to-End Distortion Analysis under Two-Hop Perturbations:}  
        We develop an analytical framework to characterize how source perturbations propagate through the semantic encoder, couple with semantic perturbations, and jointly affect end-to-end distortion through the semantic decoder.
        The derived expression decomposes distortion into source-domain, semantic-domain, and cross-hop coupling components, revealing heterogeneous codec sensitivities and explaining the limitation of capacity-based placement for DNN-based semantic relaying.
        \item \emph{Distortion-Aware UAV Placement Strategy:}
        Based on the distortion analysis, we derive a closed-form UAV placement policy for aerial SRC systems, explicitly connecting the UAV placement with device transmit power, geometry coordinates, and DNN sensitivity factors.  
        Therefore, the proposed method provides not only a low-complexity placement solution, but also an interpretable guideline for how the UAV should shift toward the SD or the BS under different source- and semantic-link sensitivities.
        \item \emph{Numerical Validation:}  
        Numerical results validate the accuracy of the proposed closed-form solution and show that it closely tracks the empirical exhaustive-search optimum.
        Compared with capacity-based and curve-fitting schemes, the proposed method determines UAV positions with lower distortion.
    \end{itemize}


\section{System Model and Problem Formulation}\label{system_model}

\subsection{System Model}
    As illustrated in Fig.~\ref{fig:system_model}, we consider an aerial SRC system where the resource-constrained SD transmits the sensory data (e.g., environmental images) to the BS with the assistance of the UAV. 
    Besides extending the data acquisition coverage of the BS via cooperative transmission \cite{7572068}, the UAV acts as an edge server to offer semantic encoding services for the SD with restricted computing and storage capacities, thereby enhancing spectrum efficiency through the UAV-BS semantic link. 
    Specifically, the well-trained JSCC encoder and decoder are deployed at the UAV and the BS, respectively.

    \begin{figure}[!t]
        \centering
        \includegraphics[scale=0.65]{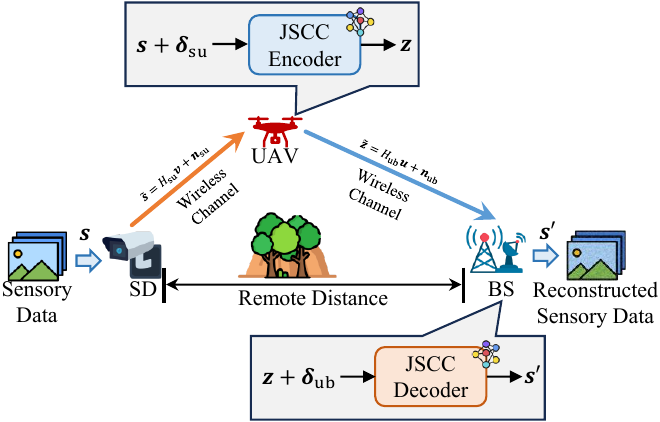}
        \caption{\small{System model of the UAV-aided SRC system.}}
        \label{fig:system_model}
    \end{figure}

    Suppose that the UAV hovers at a fixed height $\ell$ and its position is given by $\bm{u}=(x, y, \ell)$.
    We assume that the ground SD and BS are located on a two-dimensional plane with zero altitude. 
    Then, the 3D coordinates of SD, and BS are represented as $(x_\mathrm{sd}, y_\mathrm{sd}, 0)$ and $(x_\mathrm{bs}, y_\mathrm{bs}, 0)$.
    The transmission protocol of the considered UAV-assisted semantic relay system comprises the following two phases:
    \subsubsection{Sensory data transmission from SD to UAV} 
        Without the capability of performing heavy-workload DNN-based semantic encoding, the SD employs the uncoded analog transmission paradigm \cite{jakubczak2011cross} to send the source data $\bm{s}$ with dimension $D_\mathrm{su}$ to the UAV\footnote{
            Note that the non-differentiable nature of digital transmissions (e.g., quantization) severely hinders tractable end-to-end distortion analysis to facilitate UAV placement optimization. 
            We leave the extension to discrete digital architectures for future work.}
        , where the transmit signal of the SD is given by
        \vspace{-0.5em}
        \begin{align}
            \bm{v} = \sqrt{P_\mathrm{su}}\frac{\bm{s}}{\lVert\bm{s}\rVert}.
        \end{align}

        \vspace{-0.5em}
        \noindent
        Notably, $P_\mathrm{su}$ is the transmit power of the SD. 
        Suppose that the wireless channel from the SD to the UAV is dominated by the line-of-sight (LOS) link\cite{7572068}. 
        Accordingly, the channel gain follows the free-space path loss model, i.e., 
        \vspace{-0.5em}
        \begin{equation}
            H_\mathrm{su}(d_\mathrm{su}) = A (\frac{c}{4 \pi f_c d_\mathrm{su}})^\alpha,
        \end{equation}

        \vspace{-0.5em}
        \noindent
        where $A$ denotes the antenna gain, $f_c$ is the carrier frequency, $c$ is the speed of light, and $\alpha=2$ stands for the path loss exponent. 
        In addition, $d_\mathrm{su}=\sqrt{\left(x_\mathrm{sd}-x\right)^2+\left(y_\mathrm{sd}-y\right)^2+\ell^2}$ is the Euclidean distance between the SD and the UAV. 
        Then, the received signal $\tilde{\bm{s}} \in \mathbb{C}^{D_\mathrm{su}}$ at the UAV is given by
        \vspace{-0.5em}
        \begin{align}
            \tilde{\bm{s}} = \sqrt{H_\mathrm{su}(d_\mathrm{su})}e^{j\theta_\mathrm{su}} \bm{v} + \bm{n}_\mathrm{su},
        \end{align}

        \vspace{-0.5em}
        \noindent
        where $\bm{n}_\mathrm{su} \sim \mathcal{N}(0, B_\mathrm{su} N_0\bm{I})$ is the additive white Gaussian noise (AWGN) and $\theta_\mathrm{su}$ is the corresponding channel phase.
        $N_0$ and $B_\mathrm{su}$ respectively denote the noise power spectral density and bandwidth of the SD-UAV link.
        Then, the UAV post-processes the received signal through coherent detection, i.e., 
        \vspace{-0.5em}
        \begin{align}
            \hat{\bm{s}} = \frac{\tilde{\bm{s}} e^{-j\theta_\mathrm{su}} \lVert \bm{s} \rVert }{\sqrt{H_\mathrm{su}(d_\mathrm{su})P_\mathrm{su}}} = \bm{s} + \bm{\delta}_\mathrm{su},
            \label{eq:source_perturbation}
        \end{align}

        \vspace{-0.5em}
        \noindent
        where $\bm{\delta}_{\mathrm{su}} \sim \mathcal{N}(0, \sigma^2_\mathrm{su} \bm{I})$ is the Gaussian-distributed perturbation with  
        \vspace{-0.5em}
        \begin{align}
            \sigma_\mathrm{su}^2 = \frac{B_\mathrm{su} N_0 \lVert \bm{s} \rVert^2}{H_\mathrm{su}(d_\mathrm{su}) P_\mathrm{su}}.
            \label{eq:sigma_su}
        \end{align}

        \vspace{-0.5em}
        \noindent
        Accordingly, the UAV leverages its onboard JSCC encoder to extract semantic features $\bm{z} \in \mathbb{R}^{D_\mathrm{ub}}$ of the detected sensory data $\hat{\bm{s}}$, i.e., $\bm{z} = \text{E}_{\bm{\phi}} (\hat{\bm{s}})$, where $\text{E}_{\bm{\phi}}(\cdot) :: \mathbb{R}^{D_\mathrm{su}} \to \mathbb{R}^{D_\mathrm{ub}}$ is DNN-based semantic encoder parameterized by $\bm{\phi}$ and $D_\mathrm{ub}$ is the number of extracted semantic symbols.
            
    \subsubsection{Semantic feature transmission from UAV to BS}
    The UAV maps the extracted semantic features $\bm{z}$ onto a continuous constellation space and transmits these semantic symbols via analog transmission to the BS.
    The corresponding transmit signal of the UAV is
    \vspace{-0.5em}
    \begin{align}
        \bm{m} = \sqrt{P_\mathrm{ub}} \frac{\bm{z}}{\lVert \bm{z} \rVert},
    \end{align}

    \vspace{-0.5em}
    \noindent
    where $P_\mathrm{ub}$ is the transmit power of the UAV. 
    Then, the received signal at the BS is given by
    \vspace{-0.5em}
    \begin{align}
        \tilde{\bm{z}} = \sqrt{H_\mathrm{ub}(d_\mathrm{ub})}e^{j\theta_\mathrm{ub}} \bm{m} + \bm{n}_\mathrm{ub},
    \end{align}

    \vspace{-0.5em}
    \noindent
    where $\theta_\mathrm{ub}$ is the corresponding channel phase and the channel coefficient $H_\mathrm{ub}(d_\mathrm{ub})$ follows the free-space path loss model and is given by
    \vspace{-0.5em}
    \begin{align}
        H_\mathrm{ub}(d_\mathrm{ub}) = A (\frac{c}{4 \pi f_c d_\mathrm{ub}})^\alpha.
    \end{align}

    \vspace{-0.5em}
    \noindent
    $d_\mathrm{ub}=\sqrt{\left(x_\mathrm{bs}-x\right)^2+\left(y_\mathrm{bs}-y\right)^2+\ell^2}$ is the Euclidean distance between the UAV and the BS. 
    Moreover, $\bm{n}_\mathrm{ub} \sim \mathcal{N}(0, B_\mathrm{ub} N_0\bm{I})$ is the AWGN, where $B_\mathrm{ub}$ is the bandwidth of the UAV-BS link.
    Then, through the coherent detection, the BS recovers the semantic features as
    \vspace{-0.5em}
    \begin{align}
        \hat{\bm{z}} = \frac{\tilde{\bm{z}} e^{-j\theta_\mathrm{ub}} \lVert \bm{z} \rVert}{\sqrt{H_\mathrm{ub}(d_\mathrm{ub}) P_\mathrm{ub}}} = \bm{z} + \bm{\delta}_\mathrm{ub},
        \label{eq:semantic_perturbation}
    \end{align}

    \vspace{-0.5em}
    \noindent
    where $\bm{\delta}_{\mathrm{ub}} \sim \mathcal{N}(0, \sigma^2_\mathrm{ub} \bm{I})$ is the Gaussian random perturbation with 
    \vspace{-0.5em}
    \begin{align}
        \sigma_\mathrm{ub}^2 = \frac{B_\mathrm{ub} N_0 \lVert \bm{z} \rVert^2}{H_\mathrm{ub}(d_\mathrm{ub}) P_\mathrm{ub}}.
        \label{eq:sigma_ub}
    \end{align} 

    \vspace{-0.5em}
    \noindent
    Then, the BS reconstructs distorted source data $\bm{s}^{\prime} \in \mathbb{R}^{D_\mathrm{su}}$ via semantic decoding following 
    \vspace{-0.5em}
    \begin{align}
        \bm{s}^{\prime} = \text{D}_{\bm{\psi}} \left( \hat{\bm{z}} \right) = \mathrm{D}_{\bm{\psi}}\left( \mathrm{E}_{\bm{\phi}}(\bm{s} + \bm{\delta}_\mathrm{su}) + \bm{\delta}_\mathrm{ub} \right), 
        \label{eq:distorted_recon}
    \end{align}

    \vspace{-0.5em}
    \noindent
    where $\text{D}_{\bm{\psi}}(\cdot) :: \mathbb{R}^{D_\mathrm{ub}} \to \mathbb{R}^{D_\mathrm{su}}$ is the JSCC decoder parameterized by $\bm{\psi}$.

\subsection{Problem Formulation}
    The goal of the considered aerial semantic relay system is to reconstruct the sensory data at the BS with minimum distortion. 
    As indicated by \eqref{eq:source_perturbation} and \eqref{eq:semantic_perturbation}, the reconstruction error is critically dependent on both the noisy sensory data fed into the JSCC encoder and the corresponding semantic features fed into the JSCC decoder, both of which are perturbed by the wireless channel fading and communication noise. 
    Such wireless perturbations of the DNN-based JSCC codecs are governed by the UAV's position, which must be judiciously optimized to minimize the end-to-end distortion.
    
    In this paper, we quantify this distortion via the mean square error (MSE) \footnote{
        In this paper, we adopt MSE as the performance metric for analytical tractability. 
        Although other perceptual metrics such as PSNR and SSIM are widely used to evaluate reconstruction quality empirically, they are analytically intractable and thus not amenable to theoretical analysis.
        Notably, metrics like PSNR and SSIM are strongly correlated with the MSE. 
        Particularly, minimizing the MSE is equivalent to maximizing the PSNR.
    }between the reconstruction produced by noise-free JSCC codecs and the actual reconstruction $\bm{s}^\prime$ affected by the wireless perturbations at both the encoder and decoder inputs, i.e., 
    \begingroup
    \small
    \vspace{-0.5em}
    \begin{align}
        \mathcal{D}_s = \lVert \underbrace{\mathrm{D}_{\bm{\psi}}\left( \mathrm{E}_{\bm{\phi}}(\bm{s} + \bm{\delta}_\mathrm{su}) + \bm{\delta}_\mathrm{ub} \right)}_\text{With two-hop perturbations} -\underbrace{\mathrm{D}_{\bm{\psi}}\left( \mathrm{E}_{\bm{\phi}}(\bm{s}) \right)}_\text{Noise-free} \rVert_2^2. 
        \label{eq:total_distortion}
    \end{align}
    \endgroup

    \vspace{-0.5em}
    \noindent
    Accordingly, we aim to minimize the expected distortion of the considered aerial SRC system by optimizing the position of the UAV, i.e., 
    \begingroup 
    \small
    \vspace{-0.5em}
    \begin{align}
        \textbf{(P1):} \quad \min_{\bm{u}}& \quad \mathbb{E} \left[ \mathcal{D}_s \right], \\
        \text{s.t.}& \quad x \in \left[ x_\mathrm{sd}, x_\mathrm{bs} \right], y \in \left[ y_\mathrm{sd}, y_\mathrm{bs} \right],
    \end{align}
    \endgroup
    
    \vspace{-0.5em}
    \noindent
    where the expectation is taken with respect to input source data and random communication noise at both the SD-UAV and UAV-BS links.
    
    The challenges in solving Problem (P1) primarily stem from the difficulty of explicitly characterizing the impact of wireless channel fading and communication noise through the nonlinear encoder-decoder cascade. 

\section{Distortion Analysis and Closed-Form Solution Derivation}\label{proposed_method}
    To solve (P1), the key difficulty is that the objective $\mathbb{E}[D_s]$ is defined through the nonlinear encoder-decoder composition $\mathrm{D}_{\bm{\psi}}(\mathrm{E}_{\bm{\phi}}(\bm{s}+\bm{\delta}_\mathrm{su})+\bm{\delta}_\mathrm{ub})$, and thus does not directly reveal how the UAV position affects the end-to-end distortion. 
    Since the UAV position determines the two perturbation variances $\sigma_\mathrm{su}^2$ and $\sigma_\mathrm{ub}^2$ through the distances $d_\mathrm{su}$ and $d_\mathrm{ub}$, respectively, we first need an explicit distortion expression in terms of these two variances. 
    The following theorem provides such an analytical surrogate by applying a second-order Taylor expansion to the encoder-decoder cascade. 
    This surrogate will then serve as the basis for deriving the closed-form UAV placement policy.

    \begin{theorem}
        For the considered SRC system, the expected end-to-end reconstruction distortion is given by
        \vspace{-0.5em}
        \begin{align}
            \mathbb{E}\left[ \mathcal{D}_s \right] &= C_1\sigma_\mathrm{su}^4 + C_2 \sigma_\mathrm{ub}^4 + C_3\sigma_\mathrm{su}^2 + C_4\sigma_\mathrm{ub}^2\notag \\
            & + (C_5 + C_6)\sigma_\mathrm{su}^2 \sigma_\mathrm{ub}^2,
            \label{eq:proposition_1}
        \end{align}

        \vspace{-0.5em}
        \noindent
        where 
        \vspace{-0.5em}
        \begin{equation}
            \resizebox{0.75\columnwidth}{!}{$
            \displaystyle
            \left\{
            \begin{aligned}
            C_1 &= \frac{1}{4} \mathbb{E}_s \left[\sum^{D_{\mathrm{su}}}_{i=1} \mathrm{Tr}\left[ \bm{\mathcal{M}}^{(i)}_{s} \right]^2 + 2 \lVert \bm{\mathcal{M}}_{s} \rVert_{\mathrm{F}}^2\right], \\[1mm]
            C_2 &= \frac{1}{4} \mathbb{E}_s \left[\sum^{D_{\mathrm{su}}}_{i=1} \mathrm{Tr}\left[ \bm{H}_{s, \mathrm{D}}^{(i)} \left(\mathrm{E}_{\bm{\phi}}(\bm{s})\right) \right]^2 + 2\lVert \bm{H}_{s, \mathrm{D}} \left(\mathrm{E}_{\bm{\phi}}(\bm{s})\right) \rVert_{\mathrm{F}}^2\right], \\[1mm]
            C_3 &= \mathbb{E}_s \left[\mathrm{Tr} \left[ \left( \bm{J}_{s, \mathrm{D}} \left(\mathrm{E}_{\bm{\phi}}(\bm{s})\right) \bm{J}_{s, \mathrm{E}}(\bm{s}) \right)^\top \left(\bm{J}_{s, \mathrm{D}} \left(\mathrm{E}_{\bm{\phi}}(\bm{s})\right) \bm{J}_{s, \mathrm{E}}(\bm{s})\right) \right]\right], \\[1mm]
            C_4 &= \mathbb{E}_s \left[\mathrm{Tr} \left[ \bm{J}_{s, \mathrm{D}}^\top\left(\mathrm{E}_{\bm{\phi}}(\bm{s})\right) \bm{J}_{s, \mathrm{D}}\left(\mathrm{E}_{\bm{\phi}}(\bm{s})\right) \right]\right], \\[1mm]
            C_5 &= \mathbb{E}_s \left[\lVert \bm{J}_{s, \mathrm{E}}^\top(\bm{s}) \bm{H}_{s, \mathrm{D}}\left(\mathrm{E}_{\bm{\phi}}(\bm{s})\right) \rVert_\mathrm{F}^2\right], \\[1mm]
            C_6 &= \frac{1}{2} \mathbb{E}_s \left[\sum^{D_{\mathrm{su}}}_{i=1} \mathrm{Tr}\left[ \bm{\mathcal{M}}^{(i)}_{s}\right] \mathrm{Tr}\left[ \bm{H}_{s, \mathrm{D}}^{(i)} \left(\mathrm{E}_{\bm{\phi}}(\bm{s})\right) \right]\right].
            \end{aligned}
            \right.
            $}
            \label{eq:algebra_constants_1}
        \end{equation}

        \vspace{-0.5em}
        \noindent
        Notably, $\bm{J}_{s, \mathrm{E}}(\bm{s})$ and $\bm{J}_{s, \mathrm{D}}\left(\mathrm{E}_{\bm{\phi}}(\bm{s})\right)$ are the Jacobian matrices of encoder and decoder with respect to $\bm{s}$ and $\mathrm{E}_{\bm{\phi}}(\bm{s})$, respectively.
        Moreover, $\bm{H}_{s, \mathrm{E}}(\bm{s})$ and $\bm{H}_{s, \mathrm{D}}\left(\mathrm{E}_{\bm{\phi}}(\bm{s})\right)$ being the Hessian matrices of encoder and decoder with respect to $\bm{s}$ and $\mathrm{E}_{\bm{\phi}}(\bm{s})$, respectively.
        In addition, $\bm{\mathcal{M}}_s = \left[ \bm{\mathcal{M}}^{(1)}_s, \cdots, \bm{\mathcal{M}}^{(i)}_{s}, \cdots, \bm{\mathcal{M}}^{(D_\mathrm{su})}_s \right]$ and $\bm{\mathcal{M}}^{(i)}_{s} = \sum^{D_\mathrm{ub}}_{j=1} ( \bm{J}_{s, \mathrm{D}}^{(i, j)}\left(\mathrm{E}_{\bm{\phi}}(\bm{s})\right) \bm{H}_{s, \mathrm{E}}^{(j)} (\bm{s})) + \bm{J}_{s, \mathrm{E}}^\top(\bm{s}) \bm{H}_{s, \mathrm{D}}^{(i)}\left(\mathrm{E}_{\bm{\phi}}(\bm{s})\right) \bm{J}_{s, \mathrm{E}}(\bm{s})$. 
    \end{theorem}

    \begin{proof}
        We first apply Taylor expansion on $\mathrm{E}_{\bm{\phi}}(\bm{s} + \bm{\delta}_\mathrm{su})$ with respect to $\bm{s}$
        \begingroup
        \small
        \vspace{-0.5em}
        \begin{equation}
            \resizebox{1.0\linewidth}{!}{$
            \mathrm{E}_{\bm{\phi}}(\bm{s} + \bm{\delta}_\mathrm{su}) = \mathrm{E}_{\bm{\phi}}(\bm{s}) + \bm{J}_{s, \mathrm{E}}(\bm{s}) \bm{\delta}_\mathrm{su} + \frac{1}{2} \bm{\delta}^\top_\mathrm{su} \bm{H}_{s, \mathrm{E}} \bm{\delta}_\mathrm{su} + \mathcal{O}\left(\lVert \bm{\delta}_\mathrm{su} \rVert_2^3\right), 
            $}
        \end{equation}
        \endgroup

        \vspace{-0.5em}
        \noindent
        where the remainder contains higher-order terms that become negligible when the perturbation norm is sufficiently small.
        Thus, we leverage second-order approximation on $\mathrm{E}_{\bm{\phi}}(\bm{s} + \bm{\delta}_\mathrm{su})$, i.e., 
        \vspace{-0.5em}
        \begin{align}
            \mathrm{E}_{\bm{\phi}}(\bm{s} + \bm{\delta}_\mathrm{su}) \approx \mathrm{E}_{\bm{\phi}}(\bm{s}) + \bm{J}_{s, \mathrm{E}} \bm{\delta}_\mathrm{su} + \frac{1}{2} \bm{\delta}^\top_\mathrm{su} \bm{H}_{s, \mathrm{E}} \bm{\delta}_\mathrm{su}. 
            \label{eq:enc_so}
        \end{align}
        By substituting \eqref{eq:enc_so} into \eqref{eq:distorted_recon} and applying second-order Taylor expansion on the decoder, we have
        \vspace{-0.5em}
        \begingroup
        \small
        \begin{align}
            &\quad \mathrm{D}_{\bm{\psi}}\left( \mathrm{E}_{\bm{\phi}}(\bm{s} + \bm{\delta}_\mathrm{su}) + \bm{\delta}_\mathrm{ub} \right) - \mathrm{D}_{\bm{\psi}}\left( \mathrm{E}_{\bm{\phi}}(\bm{s}) \right) \notag \\
            &\approx \mathrm{D}_{\bm{\psi}}( \mathrm{E}_{\bm{\phi}}(\bm{s}) + \bm{J}_{s, \mathrm{E}} \bm{\delta}_\mathrm{su} + \frac{1}{2} \bm{\delta}^\top_\mathrm{su} \bm{H}_{s, \mathrm{E}} \bm{\delta}_\mathrm{su} + \bm{\delta}_\mathrm{ub}) - \mathrm{D}_{\bm{\psi}}\left( \mathrm{E}_{\bm{\phi}}(\bm{s}) \right), \notag \\
            &\approx \underbrace{\bm{J}_{s, \mathrm{D}} \bm{J}_{s, \mathrm{E}} \bm{\delta}_\mathrm{su} + \bm{J}_{s, \mathrm{D}} \bm{\delta}_\mathrm{ub}}_{\Delta s_1} \notag \\
            &\quad + \underbrace{
                \begin{aligned}
                    \frac{1}{2} \bigg( &\bm{J}_{s, \mathrm{D}} \bm{\delta}^\top_\mathrm{su} \bm{H}_{s, \mathrm{E}} \bm{\delta}_\mathrm{su} + \bm{\delta}_\mathrm{su}^\top \bm{J}_{s, \mathrm{E}}^\top \bm{H}_{s, \mathrm{D}} \bm{J}_{s, \mathrm{E}} \bm{\delta}_\mathrm{su} \\
                    &+ \bm{\delta}_\mathrm{su}^\top \bm{J}_{s, \mathrm{E}}^\top \bm{H}_{s, \mathrm{D}} \bm{\delta}_\mathrm{ub} + \bm{\delta}_\mathrm{ub}^\top \bm{H}_{s, \mathrm{D}} \bm{J}_{s, \mathrm{E}} \bm{\delta}_\mathrm{su} + \bm{\delta}_\mathrm{ub}^\top \bm{H}_{s, \mathrm{D}} \bm{\delta}_\mathrm{ub} \bigg)
                \end{aligned}
            }_{\Delta s_2}.
        \end{align}
        \endgroup

        \vspace{-0.5em}
        \noindent
        By defining $\Delta \bm{s} = \mathrm{D}_{\bm{\psi}}\left( \mathrm{E}_{\bm{\phi}}(\bm{s} + \bm{\delta}_\mathrm{su}) + \bm{\delta}_\mathrm{ub} \right) - \mathrm{D}_{\bm{\psi}}\left( \mathrm{E}_{\bm{\phi}}(\bm{s}) \right)$, the expected end-to-end distortion can be expressed as
        \vspace{-0.5em}
        \begin{align}
            \mathbb{E}\left[ \lVert \Delta \bm{s} \rVert_2^2 \right] \overset{(a)}{=} \mathbb{E}\left[ \Delta \bm{s}_1^\top \Delta \bm{s}_1 \right] + \mathbb{E}\left[ \Delta \bm{s}_2^\top \Delta \bm{s}_2 \right], 
            \label{eq:e2e_distortion}
        \end{align}

        \vspace{-0.5em}
        \noindent
        where (a) follows from the fact that the odd-order moments of a zero-mean Gaussian random variable are zero.

        Recall that $\bm{\delta}_{\mathrm{su}}$ and $\bm{\delta}_\mathrm{ub}$ are independent Gaussian variables.
        For $\mathbb{E} \left[ \Delta \bm{s}_1^\top \Delta \bm{s}_1 \right]$, we have 
        \vspace{-0.5em}
        \begin{equation}
            \resizebox{0.98\columnwidth}{!}{$
            \displaystyle
            \begin{aligned}
                &\quad \mathbb{E} \left[ \Delta \bm{s}_1^\top \Delta \bm{s}_1 \right] = \mathbb{E} \left[ \lVert \mathrm{D}_{\bm{\psi}}\left( \mathrm{E}_{\bm{\phi}}(\bm{s} + \bm{\delta}_\mathrm{su}) + \bm{\delta}_\mathrm{ub} \right) - \mathrm{D}_{\bm{\psi}}\left( \mathrm{E}_{\bm{\phi}}(\bm{s}) \right) \rVert_2^2 \right], \\
                &\approx \mathbb{E} \left[ \left( \bm{J}_{s, \mathrm{D}} \left( \bm{J}_{s, \mathrm{E}} \bm{\delta}_\mathrm{su} + \bm{\delta}_\mathrm{ub} \right) \right)^\top \left( \bm{J}_{s, \mathrm{D}} \left( \bm{J}_{s, \mathrm{E}} \bm{\delta}_\mathrm{su} + \bm{\delta}_\mathrm{ub} \right) \right)  \right], \\
                &=\mathbb{E}\left[ \bm{\delta}_\mathrm{su}^\top \left( \bm{J}_{s, \mathrm{D}} \bm{J}_{s, \mathrm{E}} \right)^\top \left(\bm{J}_{s, \mathrm{D}} \bm{J}_{s, \mathrm{E}}\right) \bm{\delta}_\mathrm{su} \right] + \mathbb{E}\left[ \bm{\delta}_\mathrm{ub}^\top \bm{J}_{s, \mathrm{D}}^\top \bm{J}_{s, \mathrm{D}} \bm{\delta}_\mathrm{ub} \right], \\
                &\overset{(b)}{=}\sigma_\mathrm{su}^2 \mathbb{E}\left[\mathrm{Tr}\left[ \left( \bm{J}_{s, \mathrm{D}} \bm{J}_{s, \mathrm{E}} \right)^\top \left(\bm{J}_{s, \mathrm{D}} \bm{J}_{s, \mathrm{E}}\right) \right] \right] + \sigma_\mathrm{ub}^2 \mathbb{E}\left[\mathrm{Tr}\left[ \bm{J}_{s, \mathrm{D}}^\top \bm{J}_{s, \mathrm{D}} \right]\right],
            \end{aligned}
            $}
            \label{eq:first_order_distortion}
        \end{equation}

        \vspace{-0.5em}
        \noindent
        where (b) follows Isserlis's Theorem. 

        In the following, we focus on deriving $\mathbb{E}\left[ \Delta \bm{s}_2^\top \Delta \bm{s}_2 \right]$. 
        We first define $\bm{P} = \frac{1}{2} \bm{J}_{s, \mathrm{D}} \bm{\delta}^\top_\mathrm{su} \bm{H}_{s, \mathrm{E}} \bm{\delta}_\mathrm{su}$, 
        $\bm{Q} = \frac{1}{2} \left( \bm{J}_{s, \mathrm{E}} \bm{\delta}_\mathrm{su} \right)^\top \bm{H}_{s, \mathrm{D}} \left(\bm{J}_{s, \mathrm{E}} \bm{\delta}_\mathrm{su}\right)$,
        $\bm{U} = \frac{1}{2} \left( \bm{J}_{s, \mathrm{E}} \bm{\delta}_\mathrm{su} \right)^\top \bm{H}_{s, \mathrm{D}} \bm{\delta}_\mathrm{ub}$, 
        $\bm{M} = \frac{1}{2} \bm{\delta}_\mathrm{ub}^\top \bm{H}_{s, \mathrm{D}} \left( \bm{J}_{s, \mathrm{E}} \bm{\delta}_\mathrm{su} \right)$, 
        and $\bm{N} = \frac{1}{2} \bm{\delta}_\mathrm{ub}^\top \bm{H}_{s, \mathrm{D}} \bm{\delta}_\mathrm{ub}$.

        According to the symmetric characteristics of Hessian matrix, we can have $\bm{U} = \bm{M}$.
        Then, we have 
        \vspace{-0.5em}
        \begin{equation}
            \resizebox{1\columnwidth}{!}{$
            \displaystyle
            \begin{aligned}
                &\mathbb{E}\left[ \Delta \bm{s}_2^\top \Delta \bm{s}_2 \right] = \mathbb{E}\left[\left( \bm{P} + \bm{Q} + 2\bm{U} + \bm{N} \right)^\top \left( \bm{P} + \bm{Q} + 2\bm{U} + \bm{N} \right) \right] \\
                &\overset{(c)}{=}\mathbb{E}_s\left[\left( \bm{P} + \bm{Q} \right)^\top \left( \bm{P} + \bm{Q} \right)\right] + 4\mathbb{E}\left[ \bm{U}^\top \bm{U}\right] + \mathbb{E}\left[ \bm{N}^\top \bm{N} \right] + 2 \mathbb{E}\left[\bm{P} + \bm{Q} \right] \mathbb{E}\left[ \bm{N} \right] \\
                &\overset{(d)}{=}\frac{\sigma_\mathrm{su}^4}{4} \mathbb{E}_s\left[ \sum^{D_\mathrm{su}}_{i=1} \mathrm{Tr}\left[ \bm{\mathcal{M}}^{(i)}_{s} \right]^2 + 2 \left\lVert \bm{\mathcal{M}}_{s} \right\rVert_\mathrm{F}^2 \right] + \frac{\sigma_\mathrm{ub}^4}{4} \mathbb{E}_s\left[ \sum^{D_\mathrm{su}}_{i=1} \mathrm{Tr}\left[ \bm{H}_{s, \mathrm{D}}^{(i)} \right]^2 + 2\left\lVert \bm{H}_{s, \mathrm{D}} \right\rVert_\mathrm{F}^2 \right] \\
                &\quad+ \sigma_\mathrm{su}^2 \sigma_\mathrm{ub}^2 \mathbb{E}_s\left[ \left\lVert \bm{J}_{s, \mathrm{E}}^\top \bm{H}_{s, \mathrm{D}} \right\rVert_\mathrm{F}^2 + \frac{1}{2} \sum^{D_\mathrm{su}}_{i=1} \mathrm{Tr}\left[ \bm{\mathcal{M}}^{(i)}_{s}\right] \mathrm{Tr}\left[ \bm{H}_{s, \mathrm{D}}^{(i)} \right] \right].
            \end{aligned}
            $}
            \label{eq:second_order_distortion}
        \end{equation}

        \vspace{-0.5em}
        \noindent        
        where (c) and (d) follow from the vanishing odd moments of zero-mean Gaussian variables and Isserlis's Theorem, respectively.
        By plugging \eqref{eq:first_order_distortion} and \eqref{eq:second_order_distortion} into \eqref{eq:e2e_distortion}, we can obtain \eqref{eq:proposition_1} in Theorem 1 and thus complete the proof. 
    \end{proof}

    From Theorem~1, we obtain the following observations:
    \begin{itemize}
        \item \textbf{Heterogeneous Sensitivity and Cross-Hop Coupling:}
        The two-hop perturbations have different effects on the end-to-end distortion due to the heterogeneous sensitivities of the semantic codecs. 
        Specifically, the impacts of the first-hop perturbation injected in the encoder input are governed by the cascaded Jacobian $\bm{J}_{s, \mathrm{D}}\bm{J}_{s, \mathrm{E}}$ and the mixed encoder-decoder curvature $\bm{\mathcal{M}}^{(i)}_s$, corresponding to coefficients $C_3$ and $C_1$, respectively.
        In contrast, the second-hop perturbation is injected directly before the decoder affects the end-to-end distortion through the decoder-only coefficients, i.e., the decoder Jacobian $\bm{J}_{s, \mathrm{D}}$ and Hessian $\bm{H}_{s, \mathrm{D}}$, corresponding to coefficients $C_4$ and $C_2$.
        More importantly, Theorem~1 further reveals that source and semantic perturbations do not trigger end-to-end distortion solely, but are coupled through $(C_5 + C_6)\sigma_\mathrm{su}^2 \sigma_\mathrm{ub}^2$.
        This coupling term characterizes how first-hop source perturbation propagates through the encoder, interacts with second-hop semantic perturbation through the decoder, and jointly shapes the end-to-end distortion.

        \item \textbf{The impact of wireless channel and radio resources on end-to-end distortion:}
        By plugging $\sigma_{\rm su}^2$ in \eqref{eq:sigma_su} and $\sigma_{\rm ub}^2$ \eqref{eq:sigma_ub} into Theorem~1, we can explicitly reveal how channel conditions and wireless resources affect the end-to-end distortion in UAV-aided SRC.
        Specifically, the distortion increases with the SD-UAV and UAV-BS propagation distances, i.e., $d_\mathrm{su}$ and $d_\mathrm{ub}$.
        However, simultaneously decreasing one of $d_\mathrm{su}$ and $d_\mathrm{ub}$ definitely increases the other one, indicating a trade-off of them to minimize distortion.
        In addition, the distortion also increases with the communication bandwidths $B_\mathrm{su}$ and $B_\mathrm{ub}$, since wider bandwidth introduces larger accumulated noise power.
        In contrast, increasing the transmit powers $P_\mathrm{su}$ and $P_\mathrm{ub}$ suppresses the corresponding perturbation variances and thus reduces the end-to-end distortion.
    \end{itemize}
    Accordingly, Theorem~1 indicates that \emph{UAV placement should account not only for wireless channel conditions, but also for codec-dependent cross-domain perturbation coupling.}
    
    In the following, we focus on deriving the optimal closed-form solution for UAV placement. 
    Without loss of generality, we set the SD as the origin, i.e., $x_\mathrm{sd}=y_\mathrm{sd}=0$. 
    Moreover, we establish a coordinate system by aligning the $x$-axis with the line segment from the SD to the BS, and denote their separation as $x_\mathrm{sb}$.
    In this regard, given the altitude, the optimal UAV placement reduces to a one-dimensional optimization problem along the $x$-axis. 
    Let $\hat{x}$ denote the horizontal position of the UAV along the $x$-axis in the new coordinate system. 
    The SD-UAV and UAV-BS distances are given by $d_\mathrm{su}=\sqrt{\hat{x}^2+\ell^2}$ and $d_\mathrm{ub}=\sqrt{(x_\mathrm{sb}-\hat{x})^2+\ell^2}$, respectively. 
    According to Theorem 1, Problem (P1) can be converted as
    \vspace{-0.5em}
    \begin{align}
        \text{(P2):} \quad \min_{\hat{x}} \ \mathbb{E}_s\left[ \lVert \Delta \bm{s} \rVert_2^2 \right], \quad \text{s.t.} \  \hat{x} \in \left[0, x_\mathrm{sb} \right], 
    \end{align}

    \vspace{-0.5em}
    \noindent
    where
    \begingroup
    \small
    \vspace{-0.5em}
    \begin{align}
        \mathbb{E}_s\left[ \lVert \Delta \bm{s} \rVert_2^2 \right] = \Gamma_1 d^4_\mathrm{su} + \Gamma_2 d^4_\mathrm{ub} + \Gamma_3 d_\mathrm{su}^2 + \Gamma_4 d_\mathrm{ub}^2 + \Gamma_5 d_\mathrm{su}^2 d_\mathrm{ub}^2.
        \label{eq:before_proposition_1}
    \end{align}
    \endgroup

    \vspace{-0.5em}
    \noindent
    Here
    \vspace{-0.5em}
    \begin{align}
        \begin{cases}
            T_{\mathrm{su}} = \frac{4 \pi f_c}{c} \sqrt{\frac{B_\mathrm{su} N_0 \lVert \bm{s} \rVert^2}{A P_\mathrm{su}}}, \quad T_{\mathrm{ub}} = \frac{4 \pi f_c}{c} \sqrt{\frac{B_\mathrm{ub} N_0 \lVert \bm{z} \rVert^2}{A P_\mathrm{ub}}}, \\
            \Gamma_1 =  \mathbb{E}_s\left[C_1 T_\mathrm{su}^4\right], \Gamma_2 = \mathbb{E}_s\left[C_2 T_\mathrm{ub}^4\right], \Gamma_3 = \mathbb{E}_s\left[C_3 T_\mathrm{su}^2\right], \\
            \Gamma_4 = \mathbb{E}_s\left[C_4 T_\mathrm{ub}^2\right], \quad \Gamma_5 = \mathbb{E}_s\left[(C_5 + C_6) T_\mathrm{su}^2  T_\mathrm{ub}^2\right]. \\
        \end{cases} \notag
    \end{align}
    
    \vspace{-0.5em}

    We derive the closed-form expression of the UAV placement in the considered aerial SRC system in the following. 

    \begin{proposition}
        The optimal UAV placement in Problem (P2) is the solution to
        the cubic equation 
        \vspace{-0.5em}
        \begin{align}
            a \hat{x}^3 + b \hat{x}^2 + c \hat{x} + d=0,
        \end{align}

        \vspace{-0.5em}
        \noindent   
        where $a = \Gamma_1 + \Gamma_2 + 4\Gamma_5$, $b = -3\left(\Gamma_2 + 2\Gamma_5\right)x_\mathrm{sb}$, $c = \Gamma_1 \ell^2 + \Gamma_2(3x_\mathrm{sb}^2 + \ell^2) + 2\left(\Gamma_3 + \Gamma_4 + \Gamma_5(x_\mathrm{sb}^2 + 2\ell^2)\right)$, and $d = -\left( \Gamma_2 (x_\mathrm{sb}^2 + \ell^2) + 2\Gamma_4 + 2\Gamma_5 \ell^2 \right)x_\mathrm{sb}$.
        By defining  $p = \frac{3ac-b^2}{3a^2}$, $p = \frac{2 b^3 - 9 a b c + 27 a^2 d}{27 a^3}$,  $\Delta=(\frac{q}{2})^2 + (\frac{p}{3})^3$ and  $\mathfrak{D}(\hat{x})=\mathbb{E}_s\left[ \lVert \Delta \bm{s} \rVert_2^2 \right]$, we consider the following three cases:
        \begin{itemize}
            \item If $\Delta > 0$, the optimal UAV placement in Problem (P2) is given by $\hat{x}^{\star} = \sqrt[3]{-\frac{q}{2} + \sqrt{\Delta}} + \sqrt[3]{-\frac{q}{2} - \sqrt{\Delta}} - \frac{b}{3a}.$
            %
            \item If $\Delta=0$, we have the following two subcases:
            \begin{itemize}
                \item If $p=q=0$, we have $\hat{x}^{\star} = -\frac{b}{3a}$.
                \item Otherwise,  $\hat{x}^{\star} = \arg \min(\mathfrak{D}(\hat{x}^{\star}_i), \mathfrak{D}(\hat{x}^{\star}_j))$, where $\hat{x}^{\star}_i = \frac{3q}{p} -\frac{b}{3a}$ and $\hat{x}^{\star}_j = -\frac{3q}{2p} -\frac{b}{3a}$.
            \end{itemize}
            %
            
            %
            \item If $\Delta<0$, we have $\hat{x}^{\star} = \arg \min(\mathfrak{D}(\hat{x}^{\star}_k))$, where $\hat{x}^{\star}_{k} = 2\sqrt{-\frac{p}{3}}\left(\frac{1}{3} \mathrm{arccos}\left( \frac{3q}{2p}\sqrt{-\frac{3}{p}} - \frac{2 \pi k}{3} \right) - \frac{b}{3a} \right), k=0, 1, 2$.
            %
        \end{itemize}
        %
    \end{proposition}


        
    \begin{proof}
        According to \eqref{eq:before_proposition_1}, by calculating the partial derivative with respect to $\hat{x}$, we have 
        \vspace{-0.5em}
        \begingroup
        \small
        \begin{align}
            &\frac{\partial \mathbb{E}\left[ \lVert \Delta \bm{s} \rVert_2^2 \right]}{\partial \hat{x}} = \left( \Gamma_1 + \Gamma_2 + 4\Gamma_5 \right) \hat{x}^3 -3 \left[ \Gamma_2 + 2\Gamma_5 \right] x_\mathrm{sb} \hat{x}^2\notag \\
            &+ \left( \Gamma_1 \ell^2 + \Gamma_2(3x_\mathrm{sb}^2 + \ell^2) + 2\left(\Gamma_3 + \Gamma_4 + \Gamma_5(x_\mathrm{sb}^2 + 2\ell^2)\right) \right) \hat{x} \notag \\
            &- \left( \Gamma_2 (x_\mathrm{sb}^2 + \ell^2) + 2\Gamma_4 + 2\Gamma_5 \ell^2 \right)x_\mathrm{sb}.
        \end{align}
        \endgroup

        \vspace{-0.5em}
        \noindent
        For brevity, we define $\frac{\partial \mathbb{E}\left[ \lVert \Delta \bm{s} \rVert_2^2 \right]}{\partial \hat{x}} = a \hat{x}^3 + b \hat{x}^2 + c \hat{x} + d$. 
        Letting $\frac{\partial \mathbb{E}\left[ \lVert \Delta \bm{s} \rVert_2^2 \right]}{\partial \hat{x}} = 0$, according to Cardano's Formula we can derive the three cases stated in Proposition 1. 
    \end{proof}

    According to the analytical expression derived in Proposition 1, one can directly obtain the optimal UAV placement with low complexity.

\section{Simulation Results}\label{simulation_results}
    In this section, we evaluate the performance of the proposed UAV placement method in aerial SRC systems.


\subsection{Simulation Settings}
    \subsubsection{Dataset and System Setups}
        Following \cite{10829532, 7572068}, we set carrier frequency $f_c$, communication bandwidth $B_\mathrm{su}$ and $B_\mathrm{ub}$, antenna gain $A$, and noise power spectral density $N_0$ to 5.8 GHz, 20 M, 0.5, and -174 dBm, respectively. 
        Moreover, we set the transmit power of SD and UAV to 200 mW and 50 mW, respectively. 
        In addition, we set UAV's altitude $\ell$ to 120 m, while the distance between SD and BS $x_\mathrm{sb}$ is set to 1000 m. 
        We evaluate two representative image datasets, i.e., MNIST ($D_\mathrm{su}=784, D_\mathrm{ub}=98$) and FashionMNIST ($D_\mathrm{su}=784, D_\mathrm{ub}=128$) utilizing DeepJSCC \cite{8723589}, both of which are pre-trained over the AWGN channel with SNRs ranging from 1 dB to 16 dB.
        The coefficients in \eqref{eq:algebra_constants_1} related to Jacobian and Hessian matrices are measured over the corresponding training dataset.
        The average performance across 100 independent trials is reported for all considered methods.
    \subsubsection{Benchmarks}
        For performance comparison, we consider the following benchmarks:
        \begin{itemize}
            \item \textbf{Exhaustive Search}: 
            This method evaluates the empirical distortion at every candidate UAV deployment position and selects the one with the minimum distortion.
            \item \textbf{CF-SQP}~\cite{11006481}: 
            This method follows a curve-fitting strategy, where the measured distortion values at sampled UAV positions are fitted by a standard quadratic polynomial.
            For fairness, CF-SQP constructs the fitting curve using the same amount of training data as our proposed method when estimating  the coefficients in \eqref{eq:algebra_constants_1}.

            \item \textbf{Shannon-Capacity-based Deployment}~\cite{10638143}: 
            This method follows the max-min Shannon capacity criterion for UAV placement.
            Specifically, the UAV position is selected by maximizing the minimum transmission capacity between the SD-UAV and UAV-BS links.

            \item \textbf{Central Placement}: 
            This baseline places the UAV at the midpoint between the SD and BS, serving as a simple geometry-based strategy without considering channel conditions or codec sensitivity.
        \end{itemize}

\subsection{Performance Evaluation}
    As illustrated in Fig.~\ref{fig:simulation_results_1}, we plot the end-to-end reconstruction MSE as a function of the distance $d_\mathrm{su}$ between the SD and the UAV under two representative datasets.
    We observe that the distortion first decreases with $d_\mathrm{su}$ due to the improved channel condition of the UAV-BS link for semantic feature transmission. 
    Beyond a certain point, the reconstruction MSE increases with $d_\mathrm{su}$ because of a larger perturbation at the input of the JSCC encoder deployed on the UAV. 
    This validates the existence of the optimal UAV placement for end-to-end distortion minimization in aerial SRC systems. 
    In addition, the optimal UAV position following the derived Theorem~1 shows close proximity to that under exhaustive search method.
    Specifically, compared to the CF-SQP, our proposed approach achieves $77.30\%$ and $79.49\%$ closer distance to the optimum obtained by the exhaustive search on the MNIST and FashionMNIST datasets, respectively.
    This validates the effectiveness of the derived closed-form distortion in Theorem 1.

    \begin{figure}[!t]
        \centering
        \includegraphics[scale=0.451]{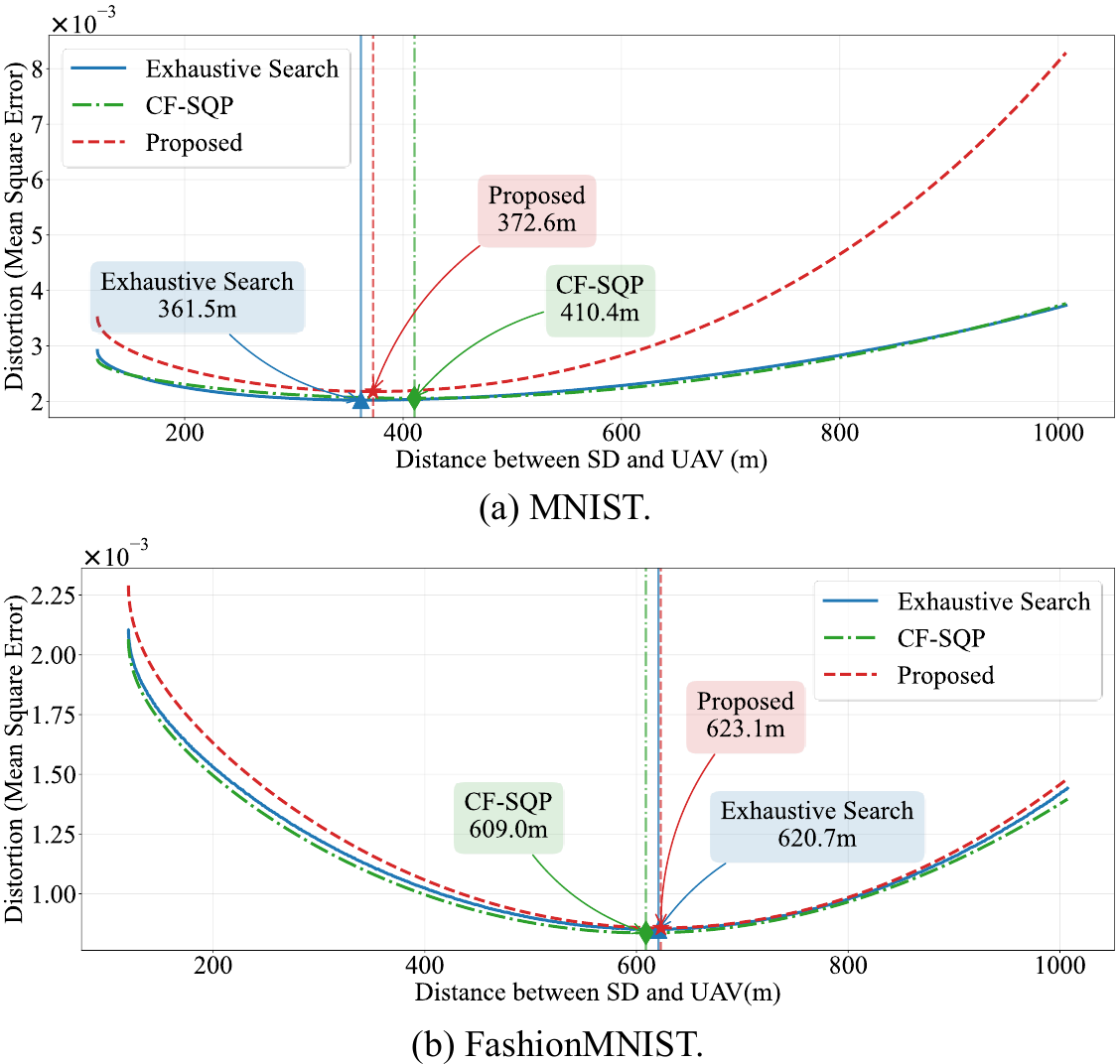}
        \caption{\small{End-to-end distortion of the UAV-aided SRC system versus the UAV's horizontal coordinate on the MNIST and FashionMNIST datasets. 
        Optimal UAV placements determined by different strategies are indicated by markers and dashed lines.}}
        \label{fig:simulation_results_1}
    \end{figure}


    In Fig.~\ref{fig:simulation_results_2}, we evaluate the achieved distortion and the optimized SD-UAV distance under a wide range of UAV transmit powers on the MNIST dataset.
    As shown in Fig.~\ref{fig:simulation_results_2}(a), the proposed method attains end-to-end distortion performance close to that of the exhaustive-search benchmark and consistently outperforms the other three baselines, i.e., achieving $52.28\%$, $52.32\%$ and $9.27\%$ lower distortion compared to central placement, Shannon-capacity-based approach, and CF-SQP methods, respectively, when the transmit power of the UAV is 0.2W. 
    This suggests the benefits of theoretically characterizing the end-to-end distortion and the closed-form UAV placement solution in the aerial SRC system.
    Fig.~\ref{fig:simulation_results_2}(b) further shows that the minimum-distortion UAV position gradually shifts toward the SD as the UAV transmit power increases, since a stronger UAV-BS link allows the UAV to move closer to the SD to mitigate source perturbation.
    Although CF-SQP can occasionally achieve accurate placement (e.g., when $P_\mathrm{ub}$ is 0.07W), its selected position shifts toward the SD more aggressively and deviates from the exhaustive-search solution in most cases, indicating the instability of empirical curve fitting method under varying transmit powers.
    In contrast, the proposed method demonstrates robustness across various UAV transmit power levels, owing to the derived analytical expression that explicitly captures the dependence of the reconstruction distortion and the resultant UAV placement on the wireless channel conditions and radio resource allocation.
    It is also worth noting that CF-SQP requires re-fitting with newly collected samples whenever the UAV transmit power changes, whereas the proposed method only needs to measure the coefficients in \eqref{eq:algebra_constants_1} related to Jacobian and Hessian matrices of the semantic codecs, showcasing the strong generalizability and fast adaptability of the proposed method across various aerial SRC system configurations.

    \begin{figure}[!t]
        \centering
        \includegraphics[scale=0.333]{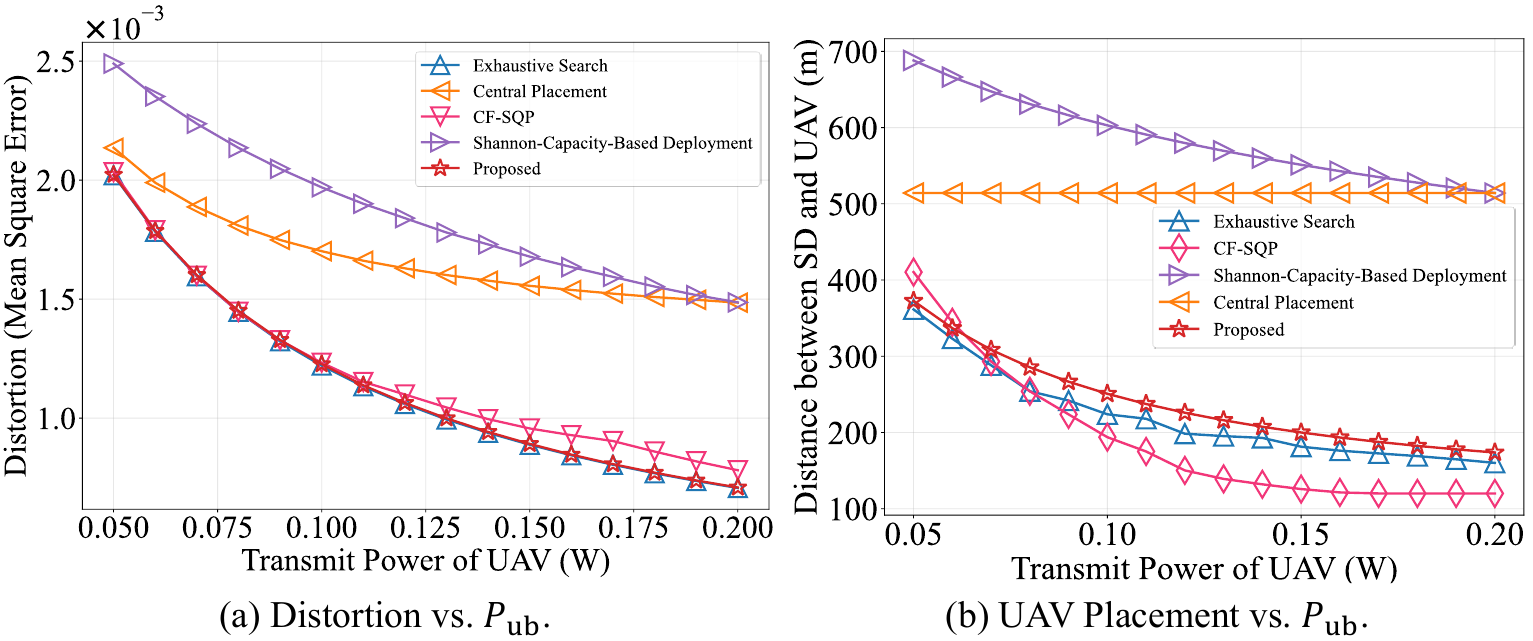}
        \caption{\small{End-to-end distortion and UAV placement determined by various methods versus the UAV's transmit power on MNIST.}}
        \label{fig:simulation_results_2}
    \end{figure}

\section{Conclusions}\label{conclusion}
    In this paper, we proposed an interpretable and low-complexity UAV placement policy by minimizing the end-to-end reconstruction distortion in aerial SRC systems under coupled two-hop wireless perturbations. We derived an analytical expression of the end-to-end distortion, explicitly characterizing how the cross-hop perturbation coupling, wireless channel conditions, and radio resource affect the reconstruction performance of the semantic codecs.
    Accordingly, we developed a closed-form UAV placement strategy with fast adaptability to different aerial SRC configurations.
    Numerical results demonstrated that the proposed distortion-aware deployment closely tracks the empirical exhaustive-search optimum and achieves lower reconstruction error than representative capacity-based and curve-fitting benchmarks.

\bibliography{ref}    
\bibliographystyle{ieeetr}
\end{document}